%% file: main.tex
\documentclass[conference]{IEEEtran}
\IEEEoverridecommandlockouts
\usepackage{cite}
\usepackage{amsmath,amssymb,amsfonts}
\usepackage{algorithmic}
\usepackage{algorithm}
\usepackage{graphicx}
\usepackage{textcomp}
\usepackage{xcolor}
\usepackage{subfigure}
\usepackage{flushend}
\usepackage{mathtools}
\usepackage{tikz}
\usetikzlibrary{positioning}
\usepackage[normalem]{ulem}
\usetikzlibrary{plotmarks}
\usepackage{algorithm}\usepackage{soul} 

\floatname{algorithm}{Algorithm}
\usepackage{algorithmic}
\usepackage{physics}
\usepackage[colorlinks=true,citecolor=blue,linkcolor=blue,urlcolor=blue]{hyperref}

\usepackage{amsthm}
\newtheorem{proposition}{Proposition}

\def\BibTeX{{\rm B\kern-.05em{\sc i\kern-.025em b}\kern-.08em
    T\kern-.1667em\lower.7ex\hbox{E}\kern-.125emX}}
\begin{document}

\title{Optimization of Quantum Systems Emulation via a Variant of the Bandwidth Minimization Problem}


\author{

\IEEEauthorrefmark{1} M. Yassine Naghmouchi,
\IEEEauthorrefmark{1} Joseph Vovrosh,
\IEEEauthorrefmark{1} Wesley da Silva Coelho,
\IEEEauthorrefmark{1} Alexandre Dauphin
\\

\\
\IEEEauthorblockA{\IEEEauthorrefmark{1}PASQAL SAS,\textit{7 rue Léonard de Vinci, 91300 Massy, France}}
}


\newcommand{\jeremie}[1]{\textcolor{red}{#1}} 

\maketitle
\thispagestyle{plain}
\pagestyle{plain}

\begin{abstract}
This paper introduces weighted-BMP, a variant of the Bandwidth Minimization Problem (BMP), with a significant application in optimizing quantum emulation. Weighted-BMP optimizes particles ordering to reduce the emulation costs, by designing a particle interaction matrix where strong interactions are placed as close as possible to the diagonal. We formulate the problem using a Mixed Integer Linear Program (MILP) and solve it to optimality with a state-of-the-art solver. To strengthen our MILP model, we introduce symmetry-breaking inequalities and establish a lower bound. Through extensive numerical analysis, we examine the impacts of these enhancements on the solver's performance. The introduced reinforcements result in an average CPU time reduction of 25.61\%. Additionally, we conduct quantum emulations of realistic instances. Our numerical tests show that the weighted-BMP approach outperforms the Reverse Cuthill-McKee (RCM) algorithm—an efficient heuristic used for site ordering tasks in quantum emulation— achieving an average memory storage reduction of 24.48\%. From an application standpoint, this study is the first to apply an exact optimization method, weighted-BMP, that considers interactions for site ordering in quantum emulation pre-processing, and shows its crucial role in cost reduction. From an algorithmic perspective, it contributes by introducing important reinforcements and lays the groundwork for future research on further enhancements, particularly on strengthening the weak linear relaxation of the MILP. 
\end{abstract}

\begin{IEEEkeywords}
Bandwidth minimization, tensor networks, quantum computing emulation.
\end{IEEEkeywords}
\newcommand{\yn}[1]{\textcolor{violet}{#1}}

\section{Introduction}
\label{section:introduction}
\input{introduction}

\section{Related work}
\label{section:relatedWork}

\input{relatedWork}

\section{Tensor Networks for the emulation of Quantum Many-Body System}
\label{section:tensornetworks}

\input{tensorNetworks}

\section{Problem statement and formulation}
\label{section:formulation}
\input{formulation}

\section{Formulation strengthening}
\label{section:strenghtening}
\input{formulationStren}

\section{Performance Evaluation}
\label{section:results}
\input{numericalResults}

\section{Conclusion and Perspectives}
\label{section:conclusion}
In this paper we have studied weighted-BMP, and its application for optimizing the tensors-to-MPS mapping in quantum emulation. We have modeled the problem using a MILP formulation, and have incorporated formulation reinforcements, such as symmetry-breaking inequalities and a theoretical lower bound. Experimental results have shown the benefits of introducing these reinforcements from an algorithmic standpoint. From an application perspective, we have shown the significant role of weighted-BMP in reducing quantum emulation costs for complex quantum systems. An interesting direction in the future is to generalize the weighted-BMP for mapping tensors to more complex structures like PEPS, TTN, and MERA. On the other hand, from an algorithmic viewpoint, this work has laid the foundation for future research, particularly in strengthening the MILP models linear relaxation, using valid inequalities.

\section*{Acknowledgment}
AD and JV acknowledge funding from the European Union under Grant Agreement 101080142 and the project EQUALITY. We also would like to acknowledge Kemal Bidzhiev for fruitful discussion.

\bibliographystyle{IEEEtran}
\bibliography{bibliography}

\section{Appendix}
\subsection{Details of MPS simulations}
\label{appendix}
We perform MPS simulations on the antiferromagnetic transverse field Ising model, applicable to Rydberg atom arrays as outlined in~\cite{henriet2020quantum}. Specifically, we analyze $N$ interacting spin $S=\frac{1}{2}$ sites within a two-dimensional amorphous solid, governed by the Hamiltonian:
\begin{equation}\label{eq:ising_hamiltonian}
H=\sum_{i<j}\frac{1}{|\mathbf{r}_i-\mathbf{r}_j|^6}\sigma^z_i\sigma^z_j+\sum_i \sigma^x_i,
\end{equation}
where the operators $\sigma^\alpha_i$ denote the Pauli matrices at the $i$-th site, positioned at $\mathbf{r}_i$. We focus on simulating the non-equilibrium dynamics starting from an initial ferromagnetic state $\ket{000\cdots0}$, observing the system up to a time $tJ=1$ with time increments of $dt = 0.01/J$. We maintain a truncation error of $10e^{-12}$ throughout the simulations to ensure energy conservation within an accuracy of $10e^{-5}$.
\end{document}

%% file: introduction.tex
Given an undirected graph $G = (V, E)$ with a set of vertices $V$ and a set of edges $E$, the Bandwidth Minimization Problem (BMP)~\cite{chinn1982bandwidth} is to find an ordering of the vertices that minimizes the maximum difference between the positions of adjacent vertices. Formally, let $\pi: V \rightarrow \{1, 2, ..., |V|\}$ be a bijection, where $\pi(u)$ represents the position of vertex $u$ in the ordering. The bandwidth, denoted as $b$, is defined as $ b = \underset{(u, v) \in E}{\max} \ |\pi(u) - \pi(v)|$.  The objective is to find a permutation $\pi$ that minimizes $b$. In addition to its graph-based formulation, the BMP can be represented using matrices. Specifically, the objective is to reorder the rows and columns of a square matrix in a way that minimizes the largest distance between the positions of non-zero elements. This ordering minimizes the distance of the nonzero components from the main diagonal of the matrix. Achieving a minimal bandwidth in $G$ is equivalent to minimizing the bandwidth in its adjacency matrix, where each vertex is assigned to the same row/column.

In the field emulation of quantum many-body systems ~\cite{thouless_quantum_1972}-involving numerous interacting sites- the BMP plays an essential role, particularly in improving the computational cost of handling Tensor Networks (TN)~\cite{orus2014practical}. Tensors, multi-dimensional arrays that generalize scalars, vectors, and matrices, are fundamental in modeling quantum many-body systems and their inherent interactions. TNs provide a mathematical framework that use these tensors, connecting them to represent quantum states and operators. Matrix Product State (MPS), a specific case of TN, are used for representing quantum systems in a one-dimensional representation. MPS decompose quantum systems, initially represented by high-ranked tensors, into sequences of interconnected tensors, where each tensor corresponds to the quantum state of an individual site within the system. In quantum emulation, the conversion of a high-rank tensor into an MPS is a common pre-processing approach. This transformation is important as it computationally facilitates emulations that would otherwise be impractical, due to the exponential growth of the Hilbert space with the increasing number of sites.


Recent studies~\cite{cataldi2021hilbert, hikihara2023automatic, legeza2003optimizing, legeza2015advanced, brabec2021massively, ali2021ordering, mate2023compressing} highlight the benefits of efficiently mapping high-ranked tensors to MPS for quantum emulations, in terms of memory storage and computational time. An efficient strategy for mapping usually involves using \textit{the interaction matrix} of the quantum many-body system, which details the strength and nature of interactions between particles in a quantum many-body system. The goal is to reorder its rows and columns in a way that minimizes the bandwidth, and then use this ordering for the MPS emulations. This ordering process is equivalent to address the BMP. Nonetheless, these studies mainly utilize the BMP in its classical definition, neglecting the magnitude of the interactions which might impact the efficiency of the ordering. Moreover, they often depend on heuristic methods such as the Reverse Cuthill-McKee (RCM) algorithm~\cite{cuthill1969reducing}, which, while reducing bandwidth, yield sub-optimal orderings that may not induce the best computational efficiency in emulations. Given the significance of incorporating site interactions in the ordering task, and the benefit of an optimal ordering through exact methods rather than heuristics, this paper introduces an exact resolution method for a BMP variant, named \textit{weighted-BMP}, that includes interactions (weights). This variant, aims at finding an ordering of the rows and columns of the interaction matrix, such that elements corresponding to stronger interactions are as close as possible to the main diagonal.

This paper presents two main contributions: an algorithmic/optimization contribution and an application contribution. From an algorithmic perspective, we model the weighted-BMP using a Mixed Integer Linear Program (MILP), which we solve to optimality. To address the computational complexity of this problem, we incorporate reinforcements such as symmetry-breaking inequalities and establish a theoretical lower bound. Furthermore, we conduct extensive numerical analysis in order to evaluate the effect of these reinforcements on the solver's performance. On the application side, we perform quantum emulations for significant physical systems, in particular, amorphous solids~\cite{zallen2008physics}, in the context of a neutral atom processor~\cite{henriet2020quantum}. We investigate the enhancement in quantum emulation efficiency through the optimal weighted-BMP ordering. Throughout this numerical study, we benchmark our weighted-BMP solutions against those obtained using the RCM method. From an application standpoint, the primary goal of this work is to underscore the role of weighted-BMP in quantum emulation. On the other hand, from an algorithmic viewpoint, our objective is to propose effective preliminary reinforcements, via symmetry-breaking inequalities and a lower bound.  

This paper is organized as follows: Section \ref{section:relatedWork} reviews related literature. In Section \ref{section:tensornetworks}, we provide background on tensor networks. Section \ref{section:formulation} introduces the problem and its MILP formulation. Section \ref{section:strenghtening} discusses formulation reinforcements. Our numerical results are presented in Section \ref{section:results}.  Section \ref{section:conclusion} concludes the paper.

%% file: relatedWork.tex
The BMP is known to be NP-hard, even for trees~\cite{garey1978complexity}. In response, researchers utilize heuristics to address near-optimal solutions that balance computational efficiency with solution quality, as explored in studies\cite{gibbs1976algorithm, dueck1995heuristic, esposito1998new, del2001heuristic}. Applications of genetic algorithms and ant-based systems in this context can be found in the work~\cite{czibula2013soft}, and a Tabu search method is employed in~\cite{marti2001reducing}. A particularly significant heuristic, which has seen extensive use in quantum emulation studies, is the RCM algorithm. First introduced by Cuthill and McKee~\cite{cuthill1969reducing}, the RCM algorithm starts with an arbitrary vertex and uses a breadth-first search to assign levels to vertices. It then sorts vertices based on these levels and further organizes them according to the number of neighbors in earlier levels. This process, repeated across all vertices, clusters them in a manner that minimizes the maximum distance between adjacent vertices in the graph, thereby reducing bandwidth, although it does not always guarantee the optimal solution. 

Another research direction in addressing the BMP focuses on developing theoretical findings, such as lower bounds, and exact solutions. This research direction is primarily motivated by the challenge posed by BMP's weak lower bounds and linear relaxations, which significantly complicates the task for exact computational methods. Initiated by Chvátal's density lower bound~\cite{chvatal1970remark}, a variety of lower bounds have been proposed for estimating the minimum bandwidth. Among the notable contributions, two methods stand out for solving this problem. The first, introduced by Del Corso~\cite{del1999finding}, is particularly effective for small to medium-sized instances. The second, developed by Caprara and Salazar~\cite{caprara2005laying}, extends the first by implementing tighter lower bounds, thus making it capable of handling larger instances. Furthermore, a branch-and-bound algorithm for the BMP is proposed in~\cite{marti2008branch}, and several formulations of the problem are studied in~\cite{coudert2016note}. 

Beyond the quantum emulation application discussed in the introduction, the BMP has significant relevance in a wide range of other applications. Particularly in the context of sparse matrices, reducing bandwidth enhances cache utilization and minimizes memory access patterns. This leads to more efficient matrix operations, such as matrix multiplication. The BMP is also known role in solving linear systems of equations like \( Ax = b \)~\cite{caprara2005laying}, where pre-processing the matrix \( A \) to minimize its bandwidth substantially reduces computational efforts. In particular, Gaussian elimination on matrices with bandwidth \( b \) can be executed in \( O(nb^2) \), offering a considerable speed advantage over the standard \( O(n^3) \) algorithm when \( b \ll n \)~\cite{lim2007fast}. The BMP also finds practical applications in diverse fields such as electromagnetic industry, chemical kinetics, numerical geophysics~\cite{esposito1999sparse}, information retrieval in hypertext systems~\cite{berry1996sparse}, and extends to graph embedding in the context of physical implementation~\cite{zhou2020quantum}.


%% file: tensorNetworks.tex
A \textit{tensor} is a multi-dimensional array that generalizes the concepts of scalars, vectors, and matrices. The rank of a tensor, indicating the number of indices or dimensions it possesses, determines its type: rank-0 for scalars, rank-1 for vectors, and rank-2 for matrices. Higher-rank tensors expand upon these basic structures. Tensors can be represented graphically through \textit{Penrose diagrams}, here, tensors are represented as nodes, with the incident edges are indices, as seen in Figure~\ref{fig:tensor_network_diagrams}. The rank is reflected by the number of edges connected to each node. Index contraction, a key operation in tensor calculus, involves summing over shared indices between tensors, analogous to the dot product in vector algebra. For instance, the matrix product \(C_{\alpha\gamma} = \sum_{\beta=1}^{D} A_{\alpha\beta} B_{\beta\gamma}\) contracts the index \(\beta\), merging dimensions of matrices \(A\) and \(B\). Graphically, a contraction of two tensors is represented as an edge connecting nodes. 

More generally, Tensor Networks (TNs) represent complex tensor operations through graphical representations with multiple nodes corresponding to tensors, and multiple edges, each representing contracted indices. This visualization simplifies complex computations and clarifies multidimensional relationships. Figure~\ref{fig:tensor_network_diagrams} illustrates TNs of various ranks. TNs include MPS, Projected Entangled Pair States
(PEPS), Tree Tensor Networks (TTN), and Multiscale Entanglement Renormalization Ansatz (MERA) which can be
defined in any dimension.

\begin{figure}[h]
    \centering
    \includegraphics[scale=0.35]{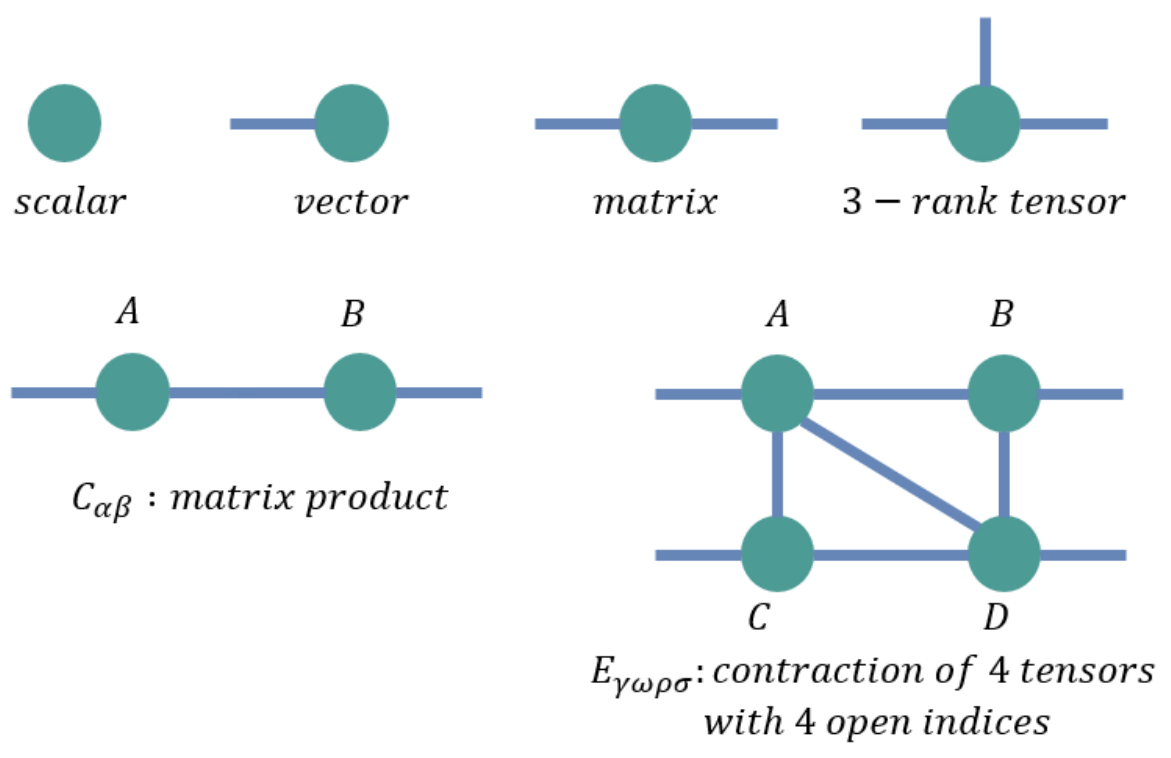}
    \caption{Graphical representation of tensor networks.}
    \label{fig:tensor_network_diagrams}
\end{figure}

In quantum computing, a \textit{qubit} is the fundamental unit of quantum information, analogous to a bit in classical computing. A qubit can exist in a superposition of two basis states, often denoted as $\ket{0}$ and $\ket{1}$. A \textit{quantum state} of a qubit is a complex linear combination of these basis states, described by a vector in a two-dimensional Hilbert space. The \textit{composite state} $\ket{\Psi}$ of a many-body quantum system represents the collective state of all sites in the system. The \textit{wavefunction} is a mathematical function that describes the quantum state, containing all possible states and their corresponding probabilities.

Emulating a quantum many-body system~\cite{thouless_quantum_1972, coleman_introduction_2015} often involves computing the systems final state after applying unitary operators~\cite{henriet2020quantum}. This process models the systems evolution under different conditions. A significant challenge in emulating these systems is the exponential growth of the state space with the number of particles. Specifically, the wavefunction of an \(N\)-qubit system is represented by a tensor of rank \(2^N\). The computational complexity of handling such a tensor grows exponentially with \(N\), making exact emulations intractable for large systems.

MPS offers a solution for the above challenge by transforming the high-rank tensor into a 1-dimensional chain TN, as seen in the Penrose diagram shown in  Figure~\ref{fig:tensor_network_decomposition}. The \textit{bond dimension} of an MPS, denoting the number of singular values retained in the tensor decomposition, is a critical factor in determining the accuracy and computational cost of the representation. Higher bond dimensions generally allow for more accurate representations but increase memory usage and computational time. The \textit{order of contractions} in MPS, crucial for emulation's computational efficiency, refers to the sequence in which tensor contractions are performed. This sequence significantly impacts the computational resources required and can affect the bond dimension. In quantum computing, this often involves mapping qubits to the nodes of an MPS, dictating how tensors associated with different qubits interact.  
An optimal contraction order is key for efficient computation, particularly in emulating entangled states in quantum systems, and can help manage the bond dimension to balance accuracy and computational cost.


\begin{figure}[h]
    \centering
    \includegraphics[scale=0.22]{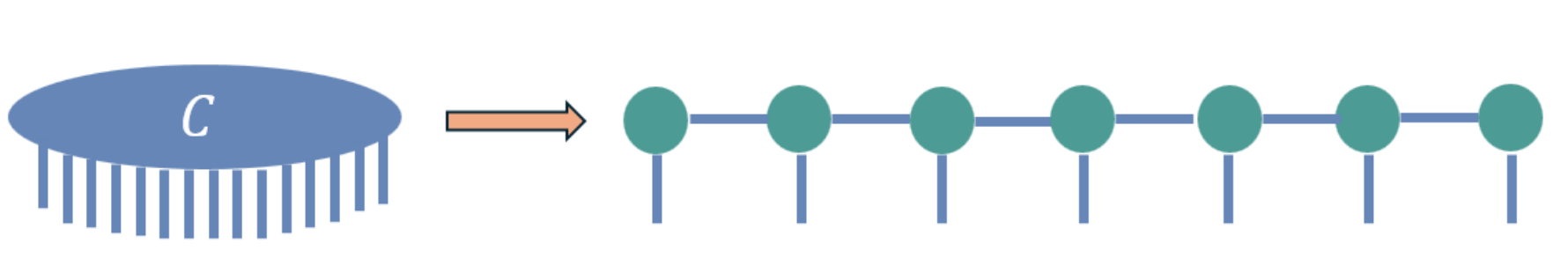}
    \caption{Decomposition of a tensor $C$ into an MPS.}
    \label{fig:tensor_network_decomposition}
\end{figure}

The interaction matrix \( U \) represents the strength and nature of interactions between different sites in a quantum many-body system. Inspired by a Rydberg atom array~\cite{henriet2020quantum}, in this work the components of \( U \) are defined by \( u_{ij} = \frac{C_6}{d_{i,j}^6} \), where \( C_6 \) is a device-dependent constant, and \( d_{i,j} \) denotes the physical distance between sites \( i, j\). It has been demonstrated that minimizing the bandwidth of the interaction matrix -- solving a BMP -- can lead to a reduction in the bond dimension of MPS~\cite{cataldi2021hilbert, hikihara2023automatic, legeza2003optimizing, legeza2015advanced, brabec2021massively, ali2021ordering, mate2023compressing}. This is a significant insight because reducing the bond dimension can greatly decrease the computational complexity and resources needed for simulating quantum many-body systems. Typically, researchers have used heuristics like RCM to minimize the bandwidth of the interaction matrix. However, these heuristics do not, generally, take into account interaction strengths, and do not guarantee optimality. Thus, they might not always yield the most efficient reduction in bond dimension. 

Our study aims to advance beyond existing heuristics by exploring the exact resolution of the weighted-BMP problem. This variant of the BMP focuses on an ordering for the rows/columns of the interaction matrix that places the largest interactions closer to the main diagonal, thereby taking into account the interactions between sites. While this method may necessitate increased computational effort, it has the promise of improving the efficiency of tensor-to-MPS mapping through further reduction of the bond dimension. In the next section, we study the weighted-BMP problem. The reader is encouraged to consult~\cite{orus2014practical} for comprehensive insights into TNs, and~\cite{rieffel2000introduction} for a foundational background in quantum computing. 

%% file: formulation.tex
In this section, we introduce the weighted-BMP problem, and present its MILP formulation.

\subsection{Problem statement}

Let $G = (V, E)$ be an undirected complete graph representing the quantum many-body system. Here, \(V\) denotes the set of qubits of size $n$, and \(E\) represents the interactions between them. Each edge \(uv \in E\) is associated with a distance \(d_{uv}\).The interaction matrix \(U\) is then defined such that \(u_{uv} = \frac{1}{d_{uv}^6}\) for all \(uv \in E\), intentionally excluding the constant \(C_6\) for model simplification. Let \(\pi: V \rightarrow \{1, 2, \ldots, n\}\) be a bijective function that assigns each vertex \(u\) to a unique ordered position \(\pi(u)\). The bandwidth \(b\) of the system is expressed as \(b = \max_{(u, v) \in E} u_{uv} |\pi(u) - \pi(v)|\). Weighted-BMP is to find the bijection \(\pi\) that minimizes the bandwidth \(b\). Such optimal ordering ensure that the vertices with the largest interactions are positioned as close as possible to each other. In terms of matrix representation, this ordering positionate the largest interactions as near to the matrix's diagonal as possible.

\subsection{MILP formulation}
We define the following decision variables:

\begin{align*}
& b: \text{a positive continuous variable representing the bandwidth,} \\
& x_v^i = \begin{cases}
      1 & \text{if the position of } v \text{ is } i, \\
      0 & \text{otherwise,}
   \end{cases} \quad \forall v \in V, \, \forall i \in \{1, \ldots, n\}.
\end{align*}
The weighted-BMP is equivalent to the following MILP
\begin{alignat}{4}
\nonumber \min \: & b && \\[6pt]
\text{s.t. } & \sum_{v \in V} x_v^i = 1  & \quad & \forall i = 1, \ldots, n, && \quad \label{constraint:ordering1} \\
& \sum_{i=1}^n x_v^i = 1  & \quad & \forall v \in V, && \quad  \label{constraint:ordering2} \\
& \sum_{i=1}^n i(x_u^i - x_v^i) \leq b \cdot d_{uv}^6  & \quad & \forall (u, v) \in E: u \neq v, && \quad \label{constraint:bandwidth} \\[6pt]
& b \geq 0, && \\
& x_v^i \in \{0, 1\}  & \quad & \forall v \in V, i = 1, \ldots, n.
\end{alignat}
Constraints~\eqref{constraint:ordering1} and~\eqref{constraint:ordering2} define an ordering of \(V\) by ensuring that a unique vertex is assigned to each position, and exactly one position is assigned to each vertex, respectively. Constraint~\eqref{constraint:bandwidth} calculates the bandwidth for each distinct pair of vertices \((u, v) \in V^2\). Here, the term $\sum_{i=1}^n i(x_u^i - x_v^i)$  is designed to compute the difference in positions between vertices
$u$ and $v$, indicating how far apart they are positioned within the ordering. The objective is to minimize the maximum bandwidth. This formulation is solved to optimality. 

While the MILP formulation of weighted-BMP is compact, involving \(n^2 + 1\) variables and \(n(n+2)\) constraints, it encounters limitations due to its weak linear relaxation and the presence of inherent symmetries that affect the efficiency of the resolution algorithm within the solver. To address these challenges, the subsequent section will detail the introduction of symmetry-breaking inequalities, and the establishment of a theoretical lower bound, aimed at fortifying the formulation. These enhancements are designed to improve the algorithmic performance of the solver by mitigating the issues associated with the initial model weaknesses.

%% file: formulationStren.tex
From constraints~\eqref{constraint:bandwidth}, it is easy to see that
\begin{equation}
\label{lower-bound}
\max_{u,v \in V^2: u \neq v} \frac{1}{{d_{uv}}^6} \leq b. 
\end{equation}
Inequality~\eqref{lower-bound} gives a theoretical lower bound that we add to the MILP formulation, in order to strengthen its linear relaxation, thereby improving the convergence of the algorithm. In addition to that, we add symmetry-breaking inequalities~\eqref{eq:symmetry_breaking} to reduce the number of nodes in the resolution tree. These inequalities are based on the following result.

\begin{proposition}
\label{symmetry-breaking1}
Let $u$ be an arbitrary vertex in $V$, then there exists an optimal solution such that 
\begin{equation}
\label{eq:symmetry_breaking}
\sum_{i=1}^{n} i x_u^i \leq \left\lceil \frac{n}{2} \right\rceil.
\end{equation}
\end{proposition}

\begin{proof}
Denote by \(x(\pi)\) the incidence vector of any weighted-BMP solution \(\pi\), where \(\pi\) is a bijective function defining an ordering: \(\pi: V \rightarrow \{1, 2, \ldots, |V|\}\), that assigns to each \(u \in V\) a position \(\pi(u)\) in the ordering.

We prove the proposition by contradiction. Assume that for any optimal solution \(\pi\), it holds that \(\sum_{i=1}^{n} ix_u^i(\pi) \geq \left\lceil \frac{n}{2} \right\rceil + 1.\) Now, define the solution \(\pi'\) such that \(\pi'(u) = n + 1 - \pi(u)\) for all \(u \in V\), which creates an inverse ordering of \(\pi\).

Evaluating the bandwidths of solutions \(\pi\) and \(\pi'\), denoted \(b(\pi)\) and \(b(\pi')\) respectively, we observe that due to the symmetry of the inverse ordering, for every pair of vertices \(u, v \in V\), the difference in positions \(|\pi'(u) - \pi'(v)|\) remains unchanged. Thus,
$b(\pi') = \max_{u,v \in V^2} \frac{1}{d_{uv}^6} |\pi'(u) - \pi'(v)| = b(\pi).$

Therefore, \( \pi' \) also constitutes an optimal solution. However, as \(\sum_{i=1}^{n} ix_u^i(\pi) \geq \left\lceil \frac{n}{2} \right\rceil + 1\), \( \pi' \) verifies 

\begin{equation}
\nonumber
\sum_{i=1}^{n} i x_u^i(\pi') = n + 1 - \sum_{i=1}^{n} x_u^i(\pi) \leq n-\left\lceil \frac{n}{2} \right\rceil \leq \left\lceil \frac{n}{2} \right\rceil,
\end{equation}

contradicting our initial assumption and concluding the proof.
\end{proof}

%% file: numericalResults.tex
In this section, we conduct numerical tests to evaluate the efficiency of the algorithmic reinforcements introduced previously, and assess the effect of weighted-BMP optimal solutions on optimizing the computational costs of quantum emulation, using a set of realistic instances.

\subsubsection{Instances description and implementation features}

We generate multiple graphs that represent amorphous solids using~\cite{juliàfarré2024amorphous}. Amorphous solids~\cite{zallen2008physics}, characterized by their lack of a crystalline structure, do not exhibit long-range order but maintain short-range order through specific bond lengths and angles akin to their crystalline counterparts. This feature leads to a well defined coordination number, \(C\), denoting the number of nearest neighbors around each site (node degree). Our generated instances simulate amorphous solids with a coordination number \(C=4\), spanning site counts of 10, 15, 20, 25, 30, 35, and 40. For each site count category, 10 distinct graphs were randomly generated to model the structure of amorphous solids. 

Our performance evaluation is based on three principal metrics: CPU time, the RCM gap to optimality, and the bond dimension. The RCM gap, expressed as a percentage, is defined by $RCM \, gap(\%) = \frac{Obj(RCM) - Opt}{Opt} \times 100,$
where \(Opt\) is the best MILP objective value, and \(Obj(RCM)\) is the RCM objective value. The bond dimension, quantifies the MPS complexity, influencing the computational resources required for quantum emulation. 


The optimization process was executed in the Python programming language. The MILP formulation is solved to optimality using CPLEX 22.1~\cite{IBMCPLEX2023}, while fixing 10 hours as a time limit. On the other hand, the emulation process was executed in the Julia programming language using ITensors package~\cite{ITensor-r0.3}. The entire experiments were conducted on a system with an AMD EPYC 7282 16-Core architecture, in 64-bit modes, with a frequency of 2800 MHz.

To provide a preliminary visual understanding before delving into detailed numerical results, Figure~\ref{fig:bandwidth_reduction} displays how a weighted-BMP solution affects the interaction matrix in a 30-site instance. The figure highlights the repositioning of highest interactions around the matrix main diagonal, which results in a reduced bandwidth.

\begin{figure}[!htb]
    \centering
    \includegraphics[width=1\linewidth]{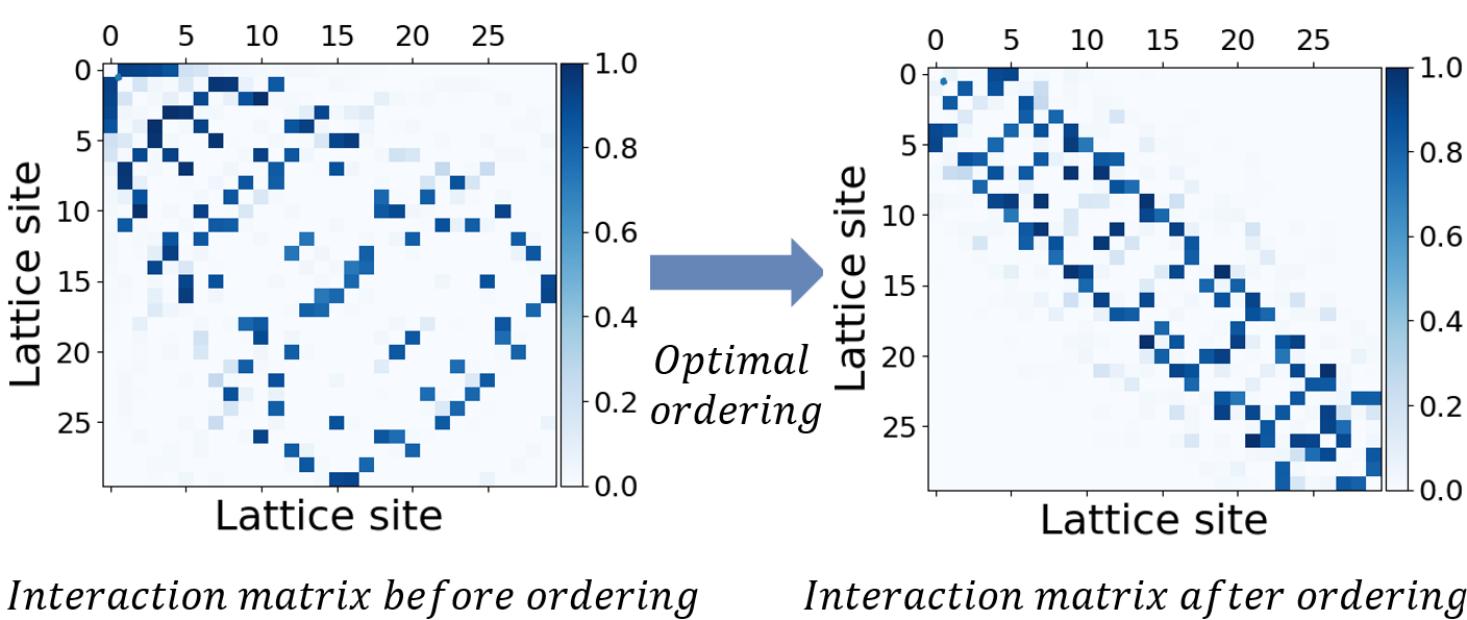}
    \caption{Effect of a solution with 30 sites on interaction matrix bandwidth}
    \label{fig:bandwidth_reduction}
\end{figure}

\subsubsection{Effect of formulation reinforcements}

Figure~\ref{fig:impact_reinforcement} illustrates the impact of formulation reinforcements on average CPU time for solving the MILP. A clear trend is observed: as the number of sites increases, so does the CPU time for both scenarios. The scenario with reinforcements consistently outperforms the one without, as evidenced by the lower average CPU times at all site numbers. The most significant difference occurs at 35 sites, where the time is nearly halved: 10782.55 seconds with reinforcements against 20442.85 seconds without. On average, incorporating reinforcements improves CPU time by 25.61\%.  Yet, at 40 sites, the limitations of reinforcements are revealed, as they fail to confirm optimality within the time limit. This is essentially caused by the weakness of the linear relaxation. Future work will focus on introducing valid inequalities to address this.

\begin{figure}[!htb]
    \centering
    \includegraphics[width=0.9\linewidth,height=0.5\linewidth]
    {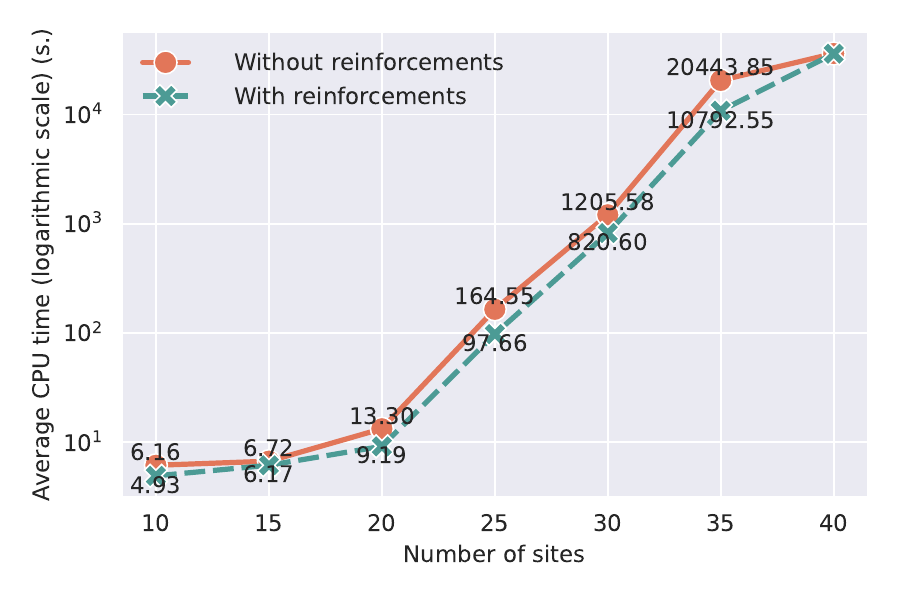}
    \caption{Effect of formulation reinforcements on Average CPU time}
    \label{fig:impact_reinforcement}
\end{figure}

The Figure~\ref{fig:rcm_gap} presents the average RCM gap percentages, for various numbers of sites. The median values of these gaps, around 18\% for 10 sites, 43\% for 30 sites, and approximately 35\% for 40 sites, indicates a consistent fail of the RCM heuristic from achieving optimality. Although some variability is observed in the data, the results identifies an overall average RCM gap to the best integer solution of 33.74\%, which underscores the general trend of optimal solutions that clearly outperform the RCM heuristic. 

\begin{figure}[!htb]
    \centering
    \includegraphics[width=0.9\linewidth,height=0.5\linewidth]{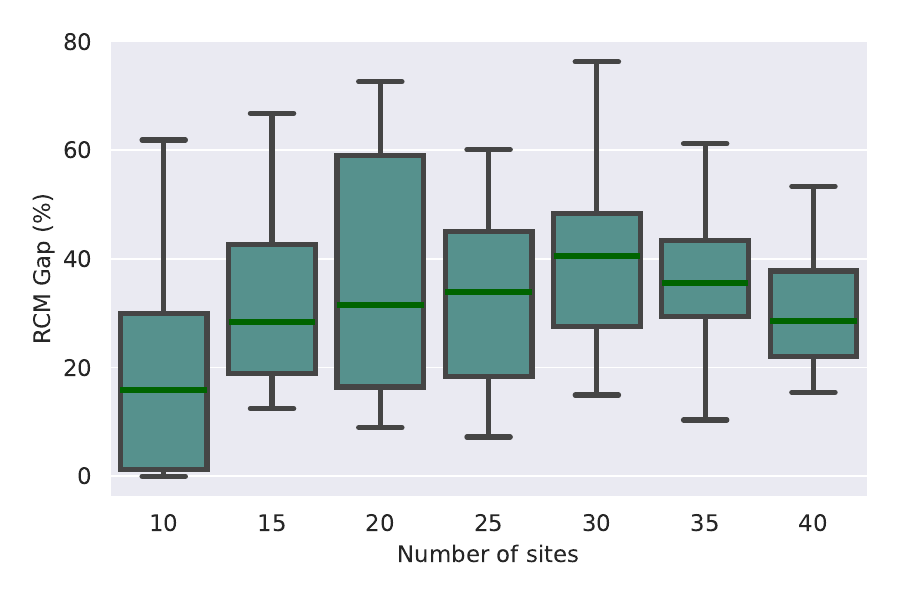}
    \caption{Average RCM gap to best weighted-BMP solution}
    \label{fig:rcm_gap}
\end{figure}

\subsubsection{Impact on quantum emulation}

Figure \ref{fig:TNsimulation} provides a comparative analysis on the average bond dimension as a function of the number of sites, between RCM and Weighted-BMP in MPS simulations. For technical details on our MPS simulations, see Appendix~\ref{appendix}. The orange curve, representing RCM, shows a consistent rise in the average bond dimension with more sites, which indicates the growing complexity and resource demands. Notably, the shaded area around the RCM curve expands with the number of sites, indicating greater variability and unpredictable performance in larger systems. In contrast, the Weighted-BMP method, represented by the green curve, not only achieves a lower average bond dimension across all site counts—with a notable difference at 35 sites where Weighted-BMP bond dimension is approximately 261.60 versus RCM's 417.10—but also displays reduced variability, as seen in the shaded area. This shows a more stable and predictable scaling performance. The overall improvement in average bond dimension offered by Weighted-BMP over RCM, quantified at 24.48\%, underscores its clear advantage in managing the complexities of quantum emulations more efficiently.

\begin{figure}[h]
    \centering
    \includegraphics[width=0.9\linewidth,height=0.5\linewidth]
    {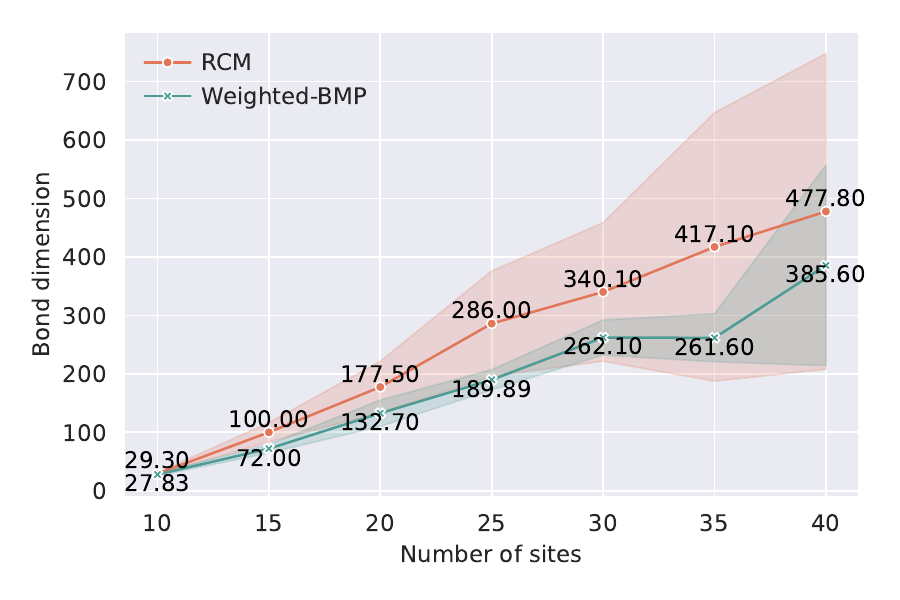} 
    \caption{Average bond dimension and variability in RCM and weighted-BMP methods.}
    \label{fig:TNsimulation}
\end{figure}